\relax
\documentclass[letterpaper,12pt]{article} %
\usepackage{times}  %
\usepackage{helvet}  %
\usepackage{courier}  %
\usepackage[hyphens]{url}  %
\usepackage{graphicx} %
\usepackage{natbib}  %
\usepackage{caption} %
\DeclareCaptionStyle{ruled}{labelfont=normalfont,labelsep=colon,strut=off} %
\usepackage{algorithm}
\usepackage{newfloat}
\usepackage{listings}
\usepackage{booktabs}
\floatstyle{ruled}
\newfloat{listing}{tb}{lst}{}
\floatname{listing}{Listing}

\usepackage{nicefrac}
\usepackage{amsmath} %
\usepackage{amsthm}
\usepackage{amsfonts}
\usepackage{algpseudocode}

\usepackage{fullpage}

\newcommand{\p}{\ensuremath{{\mathrm{P}}}}
\newcommand{\np}{\ensuremath{{\mathrm{NP}}}}
\newcommand{\conp}{\ensuremath{{\mathrm{coNP}}}}

\newcommand{\sharpp}{\ensuremath{{\mathrm{\#P}}}}

\newcommand{\sat}{{{\mathrm{sat}}}}

\newcommand{\sharplcg}{\ensuremath{\#\text{LCG}}}

\newcommand{\cohesivegroup}{\textsc{Cohesive-Group}}

\newcommand{\sharpfscg}{\ensuremath{\#\textsc{FSCG}}}

\newcommand{\sharpfsonecg}{\textsc{\#One-FSCG}}

\newcommand{\pdf}{\textsc{PD-Failure}}

\newcommand{\pdv}{\textsc{PD-Verification}}

\newcommand{\unitpdcommittee}{\textsc{Unit-PD-Committee}}

\newcommand{\sharpsetcover}{\textsc{\#Set-Cover}}

\newcommand{\sfamily}{\mathcal{S}}

\algnewcommand\And{\,\textbf{and}\,}
\algnewcommand\Or{\,\textbf{or}\,}
\algnewcommand{\LineComment}[1]{\State \(\triangleright\) #1}

\newtheorem{theorem}{Theorem}[section]
\newtheorem{proposition}[theorem]{Proposition}
\newtheorem{lemma}[theorem]{Lemma}
\newtheorem{corollary}[theorem]{Corollary}
\newtheorem{definition}{Definition}
\newtheorem{example}{Example}

\title{The Complexity of Proportionality Degree in Committee Elections}

\author{
        Łukasz Janeczko\\
        AGH University\\
        ljaneczk@agh.edu.pl
\and
        Piotr Faliszewski\\
        AGH University\\
        faliszew@agh.edu.pl
}

\begin{document}

\maketitle

\begin{abstract}
  Over the last few years, researchers have put significant effort
  into %
  understanding of the %
  notion of proportional
  representation in committee election. In particular, recently they
  have proposed the notion of proportionality degree. %
  We
  study the complexity of computing committees with a given
  proportionality degree and of testing if a given committee
  provides a particular one. This way, we complement recent
  studies that mostly focused on the notion of (extended) justified
  representation.  We also study the problems of testing if a
  cohesive group of a given size exists and of counting such groups.
\end{abstract}

\section{Introduction}

If we consider a parliamentary election where about $45\%$ voters
support party~$A$, $30\%$ support party~$B$, and the remaining $25\%$
support party~$C$, then there are well-understood ways of assigning
seats to the parties
in a proportional manner. However, if instead of
naming the supported parties the voters can approve each candidate
individually, e.g., depending on some nonpartisan agendas, then the situation
becomes less clear.
While such free-form elections are not popular in political settings,
they do appear in the context of artificial
intelligence. For example, they can be used to organize search
results~\citep{sko-lac-bri-pet-elk:c:proportional-rankings}, assure
fairness in social
media~\citep{cha-pat-gan-gum-loi:c:fair-twitter-multiwinner}, help
design online Q\&A
systems~\citep{isr-bri:c:dynamic-proportional-rankings}, or suggest
movies to watch~\citep{gaw-fal:w:multiwinner-movies}.  As a
consequence, seeking formal understanding of proportionality in
multiwinner elections is among the most active branches of
computational social choice~\citep{abc-voting-axioms}.  In this paper
we extend this line of work by analyzing the computational complexity
of one of the recent measures of proportionality, the proportionality
degree, introduced by~\citet{complexity-of-ejr-and-pjr}.

We consider the model of elections where, given a set of candidates,
each of the $n$ voters specifies which candidates he or she approves, and the
goal is to choose a size-$k$ subset of the candidates, called the
winning committee. Since the committee is supposed to represent the
voters proportionally,
\citet{justified-representation-in-approval-voting} proposed the
following requirement: For each positive integer $\ell$ and each group
of $\ell \cdot \nicefrac{n}{k}$ voters who agree on at least $\ell$
common candidates, the committee should contain at least $\ell$
candidates approved by at least one member of the group. In other
words, such a group, known as an $\ell$-cohesive group, deserves at
least $\ell$ candidates, but it suffices that a single member of the
group approves them.  If this condition holds, then we say that the
committee provides \emph{extended justified representation} (provides
EJR; if we restrict our attention to $\ell = 1$, then we speak of
providing \emph{justified representation}, JR).
The key to the success of this notion is that
committees providing EJR always exist and are selected by voting rules
designed to provide proportional representation, such as the
proportional approval voting rule (PAV) of \citet{Thie95a}.  Many
researchers followed
\citet{justified-representation-in-approval-voting} either in defining
new variants of the justified representation axioms
\citep{proportional-justified-representation,pet-pie-sko:t:fjr} or in
analyzing and designing rules that would provide committees satsifying
these
properties~\citep{complexity-of-ejr-and-pjr,san-fer-fis-bri:c:new-dhondt,bri-fre-jan-lac:c:phragmen}. Others
study these notions experimentally~\citep{BFKN19}
or analyze the restrictions that they impose~\citep{lac-sko:j:util-welfare}.

Yet, JR and EJR are somewhat unsatisfying.  After all, to provide them
it suffices that a single member of each cohesive group approves
enough committee members, irrespective of all the other voters. To
address this issue, \citet{complexity-of-ejr-and-pjr} introduced the
notion of \emph{proportionality degree}~(PD). They said that the
satisfaction of a voter is equal to the number of committee members
that he or she approves and they considered average satisfactions of the
voters in cohesive groups. More precisely, they said that
a committee has PD~$f$, where~$f$ is a function from positive integers
to nonnegative reals, if for %
each $\ell$-cohesive group, the average satisfaction of the voters in
this group is at least~$f(\ell)$. So, establishing the PD of a rule
gives quantitative understanding of its proportionality, whereas JR
and EJR only give qualitative information.  That said,
\citet{complexity-of-ejr-and-pjr} did show that if a committee
provides EJR then it has PD at least $f(\ell) = \nicefrac{\ell-1}{2}$,
so the two notions are related. Further, they showed that PAV
committees have PD $f(\ell) = \ell-1$ and
\citet{pd-of-multiwinner-rules} established good bounds on the PDs of
numerous other proportionality-oriented rules. In particular, these
results allow one to order the rules according to their theoretical
guarantees, from those providing the strongest ones to those providing
the weakest.

We extend this line of work by studying the complexity of problems
pertaining to the proportionality degree:
\begin{enumerate}
\item We show that, in general, deciding if a committee with a given
  PD exists is both $\np$-hard and $\conp$-hard and we suspect it to
  be $\np^\np$-complete (but we only show membership). Nonetheless, we
  find the problem to be $\np$-complete for (certain) constant PD
  functions. These results contrast those for JR and EJR, for which
  analogous problems are in $\p$.

\item We show that verifying if a given committee provides a given PD
  is $\conp$ complete, which is analogous to the case of EJR. We also
  provide ILP formulations that may allow one to compute PDs in
  practice (thus, one could use them to establish an empirical
  hierarchy of proportionality among multiwinner voting rules).

\item We show that many of our problems are polynomial-time solvable
  for the candidate interval and voter interval domains of
  preferences~\citep{structure-in-dichotomous-preferences} and are
  fixed-parameter tractable for the parameterizations by the number of
  candidates or the number of voters.
\end{enumerate}
We also study the complexity of finding and counting cohesive
groups.

\section{Preliminaries}
For an integer $t$, we write $[t]$ to denote the set
$\{1, \ldots, t\}$.
An election $E = (C,V)$ consists of a finite set $C$ of candidates and
a finite collection $V$ of voters.
Each voter $v \in V$ is endowed with a set $A(v) \subseteq C$ of
candidates that he or she approves.
Analogously, for each candidate $c$ we write $A(c)$ to denote the set
of voters that approve $c$; value $|A(c)|$ is known as the approval
score of $c$.
The elections considered in the $A(\cdot)$ notation will always be clear
from the context.

\paragraph{Multiwinner Voting Rules.}
A multiwinner voting rule is a function that given an election
$E = (C,V)$ and a committee size $k \leq |C|$ outputs a family of
winning committees, i.e., a family of size-$k$ subsets of $C$. (While
in practice some form of tie-breaking is necessary, theoretical
studies usually disregard this issue.) Generally, we do not focus on
specific rules, but the following three provide appropriate context
for our discussions (we assume that $E = (C,V)$ is some election and
we seek a committee of size $k$):
\begin{enumerate}
\item Multiwinner Approval Voting (AV) selects size-$k$ committees whose
  members have highest total approval score. Intuitively, AV selects committees
  of individually excellent candidates.
  
\item The Approval-Based Chamberlin--Courant rule (CC) selects those
  size-$k$ committees that maximize the number of voters who approve
  at least one member of the committee. Originally, the CC rule was
  introduced by \citet{cha-cou:j:cc} and its approval variant was
  discussed, e.g., by
  \citet{pro-ros-zoh:j:proportional-representation} and
  \citet{bet-sli-uhl:j:mon-cc}. CC selects committees of diverse
  candidates, that \emph{cover} as many voters as possible.
\item Proportional Approval Voting (PAV) selects those size-$k$
  committees $S$ that maximize the value
  $\sum_{v \in V} {w(|S \cap A(v)|)}$, where for each natural number
  $t$ we have $w(t) = \sum_{j=1}^t \nicefrac{1}{j}$. PAV selects
  committees that, in a certain sense, represent the voters
  proportionally; see, e.g., the works of
  \citet{bri-las-sko:j:multiwinner-apportionment} and
  \citet{lac-sko:j:thiele-rules}. The rule is due to
  \citet{Thie95a}.
\end{enumerate}
AV, CC, and PAV are examples of the so-called Thiele rules
\citep{Thie95a,lac-sko:j:thiele-rules}, but there are also many other
rules, belonging to other families. For more details, we point the
reader to the survey of \citet{abc-voting-axioms}.  Classifying
multiwinner rules as focused on individual excellence, diversity, or
proportionality is due to
\citet{fal-sko-sli-tal:b:multiwinner-voting}.

\paragraph{(Extended) Justified Representation.}

Let $E$ be an election with $n$ voters and let $k$ be the committee
size.  For an integer $\ell \in [k]$, called the \emph{cohesiveness
  level}, we say that a group of voters forms an $\ell$-cohesive group
if (a)~the group consists of at least $\ell \cdot \nicefrac{n}{k}$
voters, and (b)~there are at least $\ell$ candidates approved by each
member of the group. Intuitively, $\ell$-cohesive groups are large
enough to demand representation by at least $\ell$ candidates (as they
include a large-enough proportion of the voters) and they can name
these $\ell$ candidates (as there are at least $\ell$ common
candidates that they approve). Thus many proportionality axioms focus
on satisfying such demands. In particular, we are interested in the
notions of (extended) justified representation, due to
\citet{justified-representation-in-approval-voting}.
\begin{definition}
  Let $E = (C,V)$ be an election, let $k$ be a committee size, and let
  $S$ be some committee:
  \begin{enumerate}
  \item We say that $S$ provides justified representation (JR) if each 1-cohesive group
    contains at least one voter who approves at least one member of
    $S$.
  \item We say that $S$ provides extended justified representation (EJR) if for each
    $\ell \in [k]$, each $\ell$-coheseive group contains at least one
    voter that approves $\ell$ members of $S$.
  \end{enumerate}
\end{definition}
Researchers also consider other proportionality axioms, such as the
notion of \emph{proportional justified reperesentation} (PJR), due to
\citet{proportional-justified-representation}, and the recently
introduced axiom of \emph{fully justified representation} (FJR), due
to \citet{pet-pie-sko:t:fjr}. JR is the weakest of these (in the sense
that if a committee satisfies any of the other ones then it also
provides JR), followed by PJR, EJR, and FJR.  We focus on JR and EJR
as they will suffice for our purposes.
For every election and every committee size there always exists at
least one committee providing EJR (thus, also JR).
Indeed, all CC committees provide JR and all PAV committees also
provide EJR~\citep{justified-representation-in-approval-voting},
but AV committees may fail to provide (E)JR.

\paragraph{Proportionality Degree.}
Our main focus is on the notion of a proportionality degree of a
committee, introduced by~\citet{complexity-of-ejr-and-pjr}.  Let us
consider some voter $v$ and a committee $S$. We define~$v$'s
satisfaction with $S$ as $|A(v) \cap S|$, i.e., the number of
committee members that $v$ approves.
\begin{definition}
  Let $E$ be an election, let $S$ be a committee of size $k$, and let
  $f \colon [k] \rightarrow \mathbb{R}$ be a function. We say that~$S$
  has proportionality degree $f$ if for each $\ell$-cohesive group of
  voters $X$ (where $\ell \in [k]$) the average satisfaction of the
  voters in $X$ is at least $f(\ell)$.
\end{definition}
In other words, if a committee has a certain proportionality degree
$f$ for a given election, then members of the cohesive groups in this
election are guaranteed at least a certain average level of
satisfaction.  We are interested in several special types of
proportionality degree (PD) functions:
\begin{enumerate}
\item We say that $f$ is a \emph{nonzero PD} %
  if $f(\ell) > 0$ for all $\ell$.
\item We say that $f$ is a \emph{unit PD} %
  if  $f(\ell) = 1$ for all $\ell$.
\item We say that $f$ is \emph{nearly perfect} PD if
  $f(\ell) = \ell-1$ for all $\ell$.
\item We say that $f$ is a \emph{perfect} PD if $f(\ell) = \ell$ for
  all $\ell$.
\end{enumerate}
One can verify that every CC committee (or, in fact, every JR
committee) has nonzero PD, and \citet{complexity-of-ejr-and-pjr} have
shown that every PAV committee has nearly perfect PD.  It is also
known that if a committee provides EJR then, at least, it has
proportionality degree $f$ such that $f(\ell) = \nicefrac{\ell-1}{2}$
\citep{proportional-justified-representation}.  Yet there exist
elections for which no committee has unit PD or perfect
PD~\citep{complexity-of-ejr-and-pjr}.  For a detailed analysis of
proportionality degrees of various multiwinner rules, we point to the
work of \citet{pd-of-multiwinner-rules}.  Finally, we note that by
saying that a certain rule has PD $f$ we only indicate a lower bound
on its performance. Consequently, for many rules we can say that they
provide several different proportionality degrees, as in the case of
PAV, which provides both nearly perfect PD and some nonzero PD
(nonzero PD offers a stronger guarantee than the nearly perfect one
for $\ell = 1$).

\paragraph{Computational Complexity.}
We assume knowledge of classic and parameterized computational
complexity theory, including classes $\p$ and $\np$, the notions of
hardness and completeness for a given complexity class, and FPT
algorithms. Occasionally, we also refer to the $\conp$ class and to
higher levels of the Polynomial Hierarchy. Given a problem $X$ from
$\np$, where we ask if a certain mathematical object exists, we write
$\#X$ to denote its counting variant, where we ask for the number of
such objects. Counting problems belong to the class $\sharpp$ and it
is commonly believed that if a counting problem is $\sharpp$-complete
then it cannot be solved in polynomial time. We mention that
$\sharpp$-completeness is defined via Turing
reductions~\citep{val:j:permanent} and not many-one reductions, as in
the case of $\np$-completeness. A counting problem $\#X$
Turing-reduces to a counting problem $\#Y$, denoted
$\#X \le_{fp}^{T} \#Y$, if there is an algorithm that solves $\#X$ in
polynomial time, provided it has access to $\#Y$ as an oracle (i.e.,
provided that it has a subroutine for solving $\#Y$ in constant time).

\paragraph{Computational Aspects of JR, EJR, and PD.}
There are polynomial-time algorithms that given an election and a
committee size compute committees which provide
JR~\citep{justified-representation-in-approval-voting} or
EJR~\citep{complexity-of-ejr-and-pjr}. There is also a polynomial-time
algorithm that given a committee verifies if it provides JR. The same
task for EJR is
$\conp$-complete~\citep{justified-representation-in-approval-voting}.
In this paper we answer analogous questions for the case of the
proportionality degree.

\section{Finding and Counting Cohesive Groups}
\label{sec:finding-counting}

As cohesive groups lay at the heart of JR, EJR, and PD, we start our
discussion by analyzing the hardness of finding them. More precisely,
we consider the following problem.

\begin{definition}
  An instance of the $\cohesivegroup$ problem consists of an election
  $E$, a committee size $k$, and a positive integer $\ell$. We ask if
  $E$ contains an $\ell$-cohesive group.
\end{definition}

Somewhat disappointingly, this problem is $\np$-complete. This follows
via a reduction inspired by that provided by
\citet{justified-representation-in-approval-voting} to show that
testing if a given committee provides EJR is $\conp$-complete (we
include the proof for the sake of completeness, as some of our further
hardness proofs follow by reductions from \textsc{Cohesive-Group}).

\begin{theorem}\label{thm:cg-np}
  \textsc{Cohesive-Group} is $\np$-complete
\end{theorem}
\begin{proof}
  We observe that \textsc{Cohesive-Group} is in $\np$: Given an
  election $E$ with $n$ voters, committee size $k$, and cohesiveness
  level $\ell$, it suffices to nondeterministically guess a group of
  at least $\ell\cdot\nicefrac{n}{k}$ voters and check that the
  intersection of their approval sets contains at least $\ell$
  candidates.

  To show $\np$-hardness, we give a reduction from the $\np$-complete
  \textsc{Balanced-Biclique} problem~\citep{bbp-np-complete}. The
  input for \textsc{Balanced-Biclique} consists of a bipartite graph
  $G$ and a nonnegative integer $k$. The vertices of $G$ are
  partitioned into two sets, $L(G)$ and $R(G)$, and we write $E(G)$ to
  denote the set of $G$'s edges; each edge connects a vertex from
  $L(G)$ with a vertex from $R(G)$. We ask if there is a size-$k$
  subset of $L(G)$ and a size-$k$ subset of $R(G)$ such that each
  vertex from the former is connected with each vertex from the latter.
  Such two sets are jointly referred to as a $k$-biclique of $G$.

  Given an instance of \textsc{Balanced-Biclique}, we form an instance
  of \textsc{Cohesive-Group} as follows.  We construct an election
  $E'$, where $R(G)$ is the set of candidates and $L(G)$ is a
  collection of voters. A voter $\ell_i \in L(G)$ approves a candidate
  $r_j \in R(G)$ if $\ell_i$ and $r_j$ are connected in $G$.  We
  extend $E'$ by adding $\max(||L(G)| - |R(G)|,0)$ candidates not
  approved by any voter, we set the committee size to be $k' = |L|$,
  and we let the desired cohesiveness level be $\ell' = k$. This
  completes the construction.

  Note that each $\ell'$-cohesive group in our election consists of at
  least $\ell' \frac{|L|}{k'} = k \frac{|L|}{|L|} = k$ voters who
  approve at least~$k$ common candidates. Focusing on exactly $k$
  voters and $k$ candidates, we see that such a group exists if and
  only if $G$ has a $k$-biclique. This completes the proof.
\end{proof}

On the positive side,
\citet{justified-representation-in-approval-voting} gave a
polynomial-time algorithm for deciding if an election contains a
$1$-cohesive group (we refer to this variant of the problem as
\textsc{One-Cohesive-Group}): It suffices to check if there is a
candidate $c$ for whom $|A(c)| \geq \nicefrac{n}{k}$, where $n$ is the
total number of voters and $k$ is the committee size. If such a
candidate $c$ exists, then the voters from $A(c)$ form a $1$-cohesive
group; otherwise, there are no $1$-cohesive groups.

\begin{corollary}[\bf \citet{justified-representation-in-approval-voting}]
  \label{cor:1-cg-in-p}
  \textsc{One-Cohesive-Group} is in $\p$.
\end{corollary}

We complement the above results by considering the complexity of
\textsc{\#Cohesive-Group}, i.e., the problem of counting cohesive
groups.
If we had an efficient algorithm for this problem, then we could also
derive an efficient procedure for sampling cohesive groups uniformly
at random~\citep{jer-val-vaz:j:sampling-counting}, which would be
quite useful. Indeed, we could use it, e.g., to experimentally study
the distribution of cohesive groups in elections.\footnote{Formally,
  an approximate counting algorithm would suffice to obtain a nearly
  uniform sampling procedure. Our results do not preclude existence of
  such an algorithm, but we leave studies in this direction for future
  work.}  Naturally, \textsc{\#Cohesive-Group} is intractable---namely,
$\sharpp$-complete---since even deciding if a single cohesive group
exists is hard. More surprisingly, the same holds for 1-cohesive
groups.

\begin{theorem}\label{thm:count-1-cg}
  \textsc{\#One-Cohesive-Group} is $\sharpp$-complete
\end{theorem}

An intuition as to why finding a single $1$-cohesive group is easy but
counting them is hard is as follows. Using the argument from
Corollary~\ref{cor:1-cg-in-p}, for each candidate we can count (in
polynomial time) the number of $1$-cohesive groups whose members
approve this candidate. Yet, if we simply added these values, then
some groups could be counted multiple times. If we used the
inclusion-exclusion principle, then we would get the correct result,
but doing so would take exponentially many arithmetic operations.

To give a formal proof of Theorem~\ref{thm:count-1-cg}, we use the
following intermediate problem, which captures counting cohesive
groups that consist of a given number of voters.

\begin{definition}
  In the \textsc{\#Fixed-Size-Cohesive-Group} problem (the
  \textsc{\#FSCG} problem) we are given an election $E = (C, V)$, a
  committee size $k$, and two positive integers, $\ell$ and $x$. We
  ask how many $\ell$-cohesive groups that consist of exactly $x$
  voters are there in $E$.
\end{definition}

Given an instance $(E,k,\ell,x)$ of \textsc{\#FSCG}, by
$\sharplcg(E, k, \ell, x)$ we mean the number of $\ell$-cohesive
groups of size $x$ from election $E$ for committee size $k$. If we
omit parameter $x$, then we mean to total number of $\ell$-cohesive
groups, irrespective how many voters they include.  In the next
proposition we show that \textsc{\#FSCG} is computationally equivalent
to \textsc{\#CohesiveGroup} (the proof is in
Appendix~\ref{app:counting}).

\begin{proposition}\label{red:12}
  \textsc{\#Cohesive-Group} $\le_{\mathrm{T}}^{\mathrm{fp}}$ \textsc{\#FSCG}
  and
  \textsc{\#FSCG} $\le_{\mathrm{T}}^{\mathrm{fp}}$ \textsc{\#Cohesive-Group}.
\end{proposition}

We define problem \textsc{\#One-FSCG} by fixing $\ell = 1$ in the
definition of \textsc{\#FSCG}. Proposition~\ref{red:12} also holds for
the case of \textsc{\#One-FSCG} and \textsc{\#One-Cohesive-Group}
(indeed, we never modify~$\ell$ in the proposition's proof).  We use
the computational equivalence of \textsc{\#One-FSCG} and
\textsc{\#One-Cohesive-Group} to prove Theorem~\ref{thm:count-1-cg}.

\begin{proof}[Proof of Theorem~\ref{thm:count-1-cg}]
  A problem belongs to $\sharpp$ if its value can be expressed as the
  number of accepting paths of a polynomial-time nondeterministic
  Turing machine. For \textsc{\#One-Cohesive-Group} it suffices that
  such a machine guesses a group of voters, verifies if they form an
  $\ell$-cohesive group (which can be done deterministically in
  polynomial time), and accepts if so.  Hence,
  \textsc{\#One-Cohesive-Group} is in $\sharpp$.
  To show $\sharpp$-hardness of \textsc{\#One-Cohesive-Group}, we give
  a reduction from $\sharpsetcover$ to $\sharpfsonecg$.  The latter
  problem is well-known to be $\sharpp$-complete, and the former is
  computationally equivalent to \textsc{\#One-Cohesive-Group}.

  Let $(U, \sfamily, k)$ be an instance of the $\sharpsetcover$
  problem, where $U = \{u_1, \ldots, u_n\}$ is a universe,
  $\sfamily = \{S_1, \ldots, S_m\}$ is a family of subsets of $U$, and
  $k$ is a positive integer. The question is how many combinations of
  at most $k$ subsets from $\sfamily$ sum up to the universe $U$.  We
  create an instance of $\sharpfsonecg$ with an election
  $E' = (C', V')$, such that the candidates correspond to the elements
  of $U$ and the voters correspond to the elements of $\sfamily$
  (hence, we can speak both of a universe element $u_i$ and a
  candidate $u_i$, or of a set $S_j$ and voter $S_j$).  For each
  candidate $u_i$ and voter $S_j$, $u_i$ is approved by voter $S_j$ if
  element $u_i$ does not belong to the set $S_j$. Further, we extend
  $E'$ by adding $max(m-n, 0)$ new candidates not approved by
  anyone. Altogether, the number of candidates in $E'$ is at least
  $m$. We set the size of the final committee to be $k' = m$.  We also
  write $n'$ to denote the number of voters in $E'$; naturally, we
  have $n' = m$.  Due to the definition of a $1$-cohesive group, its
  size must be at least $\nicefrac{n'}{k'} = \nicefrac{m}{m} =
  1$. Thus, every group of voters that approve at least one common
  candidate is a $1$-cohesive group in our election.

  Let us consider a $1$-cohesive group of size $x' \le k$ and let us
  call it $T$. By definition, there is at least one candidate approved
  by all members of $T$. Let us call her $c'$. This means that $c'$ is
  not included in any set $S_j$ corresponding to the voters from $T$.
  Hence, the union of these sets is different from~$U$.
  On ther other hand, if a group $R$ of $x' \le k$ voters does not
  form an $1$-cohesive group, then the sets corresponding to the
  voters from this group do sum up to the universe~$U$. Indeed, for
  each candidate $c$ there is a voter in $R$ who does not approve $c$,
  which means that the corresponding set includes her.  As a
  consequence, each member of $U$ belongs to at least one set
  corresponding to a voter from $R$. 

  Above observations mean that families of subsets from $\sfamily$ sum
  up to the universe $U$ if and only if the voter groups that
  correspond to these families are not $1$-cohesive. Since for a
  positive integer $x$ there are
  ${|\sfamily| \choose x} = {m \choose x}$ size-$x$ families of sets
  from $\sfamily$, we conclude that the answer for our instance of
  \textsc{\#Set-Cover} is
  $\sum_{x=1}^{k} { ({|\sfamily| \choose x} - \sharplcg(E', 1, x)) }$.
  This completes the proof.
\end{proof}

\section{Computing a Committee with a Given PD}
\label{sec:pd}

In this section we focus on the complexity of deciding if a committee
with a given proportionality degree exists.  At first, this problem
may seem trivial as for each election there is a committee with nearly
perfect PD~\citep{complexity-of-ejr-and-pjr}.  Yet, we find that the
answer is quite nuanced. This stands in sharp contrast to analogous
decision questions for JR and EJR, which \emph{are} trivial (a
committee with the desired property always exists so the algorithm
always accepts).  Formally, we consider the following problem.
\begin{definition}
  In the \textsc{PD-Committee} problem we are given an election $E$, a
  committee size $k$, and a function $f \colon [k] \rightarrow \mathbb{Q}$,
  specified by listing its values.  We ask if $E$ has a size-$k$
  committee with proportionality degree at least $f$.
\end{definition}

We find that \textsc{PD-Committee} is both $\conp$-hard and
$\np$-hard. For the former result, we use the fact that for a given
$\ell$, the $f(\ell)$~value of a PD function is binding only if the
given election contains $\ell$-cohesive groups.

\begin{theorem}
  \textsc{PD-Committee} is both $\np$-hard and $\conp$-hard
\end{theorem}

\begin{proof}
  We will show $\np$-hardness in Theorem~\ref{thm:unit-pd} and here we
  focus on $\conp$-hardness. To this end, we give a reduction from
  \textsc{Cohesive-Group} to the complement of \textsc{PD-Committee}.
  Let $(E, k, \ell)$ be our input instance, where $E = (C,V)$ is an
  election, $k$ is the committee size, and $\ell$ is the cohesiveness
  level. The question is if there exists an $\ell$-cohesive group for
  election $E$ with committee size $k$. For convenience, we set
  $n = |V|$, and $m = |C|$.

  We create an instance of the complement of \textsc{PD-Committee} as
  follows.  Let~$s$ be the smallest integer such that $s \cdot k > m$.
  We form an election $E'$ by first copying $E$ and then adding
  (a)~$s \cdot k$ new candidates who are not approved by any voters and
  (b)~$(s-1) \cdot n$ new voters who do not approve any
  candidates. Altogether, in $E'$ we have $n' = s \cdot n$ voters, and
  $m' = m + s \cdot k$ candidates. Further, we set the committee size to be
  $k' = s \cdot k$ and we let the PD function $f$ be such that for
  $i < \ell$ we have $f(i) = 0$ and for $i \geq \ell$ we have
  $f(i) = k'$.  This completes the construction.

  Note that the minimum size of an $\ell$-cohesive group in $E'$ is
  equal to the minimum size of an $\ell$-cohesive group in $E$,
  because
  $\frac{\ell \cdot n'}{k'} = \frac{\ell \cdot s \cdot n}{s \cdot k} =
  \frac{\ell \cdot n}{k}$. Thus every  $\ell$-cohesive group from $E$
  is also an $\ell$-cohesive group for $E'$ and vice versa.
  Further, each size-$k'$ committee %
  must contain at least one new candidate, because
  $k' = s \cdot k > m$. Yet, the new candidates are not approved by
  any voter and, so, if $E'$ has some $\ell$-cohesive group, then its
  average satisfaction must be strictly below $f(\ell) = k'$.  This
  means that if $E'$ has a committee with PD $f$ then there are no
  $\ell$-cohesive groups in $E'$ (and, thus, there are no cohesive
  groups in $E$).  In other words, the answer for the
  \textsc{PD-Committee} instance is ``yes'' if and only if the answer
  for the \textsc{Cohesive-Group} instance is ``no.''
  Since, by Theorem~\ref{thm:cg-np}, the latter is $\np$-complete, the
  former is $\conp$-hard.
\end{proof}

Since \textsc{PD-Committee} is both $\np$-hard and $\conp$-hard, it is
unlikely that it is complete for either of these classes (we would
have $\np = \conp$ if it were). Indeed, we suspect that it is complete for
$\np^\np$ and we show that it belongs to this class. An $\np^\np$-hardness result
remains elusive, unfortunately.

\begin{theorem}
  \textsc{PD-Committee} is in $\np^\np$.
\end{theorem}
\begin{proof}
  Consider an instance $(E,k,f)$ of \textsc{PD-Committee}. It is a
  ``yes''-instance exactly if there exists a size-$k$ committee such
  that for every $\ell \in [k]$, every $\ell$-cohesive group has
  average satisfaction at least $f(\ell)$. We can verify that this
  holds by first nondeterministically guessing the committee and then
  asking the oracle if there is a cohesive group for which the
  constrained implied by the PD function is failed (since computing an
  average satisfaction of a given cohesive group can be done in
  polynomial time, this task belongs to $\np$). We accept if the
  oracle answers ``yes'' and we reject otherwise.\footnote{Note that
    this way our nondeterministic machine makes only a single query to
    the oracle. While one could worry that this might mean that our
    problem is somehow ``easy'' for $\np^\np$, this is not the
    case. Indeed, it is well-known that every problem in $\np^\np$ can
    be solved by a nondeterministic machine with a single oracle
    call.}
\end{proof}

While \textsc{PD-committee} seems very hard in general, for some
classes of PD functions it is significantly easier. As an extreme
example, for nearly perfect ones it is trivially in $\p$ because PAV
winning committees always have nearly perfect PD.  We consider the
following restricted variants of \textsc{PD-Committee}: In
\textsc{Constant-PD-Committee} we require the desired PD functions to
be constant, in \textsc{Unit-PD-Committee} we require them to take
value $1$ for each argument, and in \textsc{Perfect-PD-Committee} we
require them to be perfect.
We find that both \textsc{Constant-PD-Committee} and
\textsc{Unit-PD-Committee} are $\np$-complete and, thus, likely much
easier than the general variant. To establish these results, it
suffices to show membership in $\np$ for the former and $\np$-hardness
for the latter.

\begin{theorem}\label{thm:const-pd-np}
  \textsc{Constant-PD-Committee} is in $\np$.
\end{theorem}

\begin{proof}
  Consider an instance $(E,k,f)$ of \textsc{Constant-PD-Committee},
  where $E = (C,V)$ is an election, $k$ is the committee size, and $f$
  is a constant PD function. Since $f$ is a constant function, there
  is a value $x$ such that for each $\ell \in [k]$ we have
  $f(\ell) = x$. To show that \textsc{Constant-PD-Committee} is in
  $\np$, we give a polynomial-time algorithm that given such an
  instance and size-$k$ committee $W$ verifies if~$W$ has PD $f$.

  Let $n = |V|$ be the number of voters.  For each candidate
  $c \in C$, we define $\sat(c)$ to be the average satisfaction of
  $\lceil\frac{n}{k}\rceil$ members of $A(c)$ that are least satisfied
  with $W$; if $A(c)$ contains fewer than $\frac{n}{k}$ voters then we
  set $\sat(c) = +\infty$. We set $y = \min_{c \in C}\sat(c)$. If
  $y = +\infty$ then election $E$ has no cohesive groups and $W$ has
  PD $f$ trivially. Otherwise, $y$ is the smallest average
  satisfaction that a $1$-cohesive group from~$E$ has for~$W$ (indeed,
  every $1$-cohesive group must have at least
  $\lceil \frac{n}{k} \rceil$ members and for each $c \in C$, each
  $1$-cohesive group whose members approve $c$ has satisfaction at
  least $\sat(c)$).  For each $\ell \in [k]$, each $\ell$-cohesive
  group also has satisfaction at least $y$ (each such group also is a
  $1$-cohesive group and, so, also has average satisfaction at least
  $y$).  Thus, if $y \geq x$ then we accept and otherwise we reject.
  This algorithm runs in polynomial time.
\end{proof}

\begin{theorem}\label{thm:unit-pd}
    $\unitpdcommittee$ is $\np$-hard
\end{theorem}

\begin{proof}
  We give a reduction from a variant of the classic \textsc{X3C}
  problem, which we call \textsc{RX3C} and which is well-known to be
  $\np$-complete~\citep{rx3c-np-complete}: An instance of
  \textsc{RX3C} consists of a universe set
  $U = \{u_1, u_2, ... , u_{3k}\}$ and a family
  $\sfamily = \{S_1, S_2, ... , S_{3k}\}$ of size-$3$ subsets of $U$,
  each element from $U$ belongs to exactly three sets from $\sfamily$,
  and we ask if there exist $k$ subsets from $\sfamily$ which sum up
  to the universe $U$.

  We form an instance of \textsc{Unit-PD-Committee} with an election
  $E$, committee size $k$, and unit PD function. We let the sets from
  $\sfamily$ be the candidates in $E$, and we let the universe
  elements be the voters.  A voter $u_i$ approves a candidate $S_j$ if
  $u_i \in S_j$. This completes the construction.

  We note that all cohesive groups in $E$ contain exactly three voters
  and have cohesiveness level one. This holds because each candidate
  is approved by exactly three voters and this is also the lower bound
  on the size of $1$-cohesive groups in $E$ (indeed,
  $\nicefrac{3k}{k}=3$).

  It is clear that if there exist $k$ subsets from $\sfamily$ which
  sum up to $U$, then the corresponding candidates form a committee
  which has average satisfaction at least $1$. Indeed, for each voter
  there is at least one candidate in the committee that he or she
  approves (in fact, exactly one). Otherwise the selected sets would
  not sum up to $U$. As a consequence, the average satisfaction of
  each \hbox{($1$-)}cohesive group with the committee is at least $1$.

  Next, let us show that if there exists a committee $W$ of size $k$
  such that each cohesive group has average satisfaction at least $1$,
  then there is  a collection $T$ of $k$ sets from $\sfamily$ that sum
  up to~$U$ (i.e., there is an exact cover of $U$).
  Let $B$ be the sum of the total satisfactions of all the $3k$
  $1$-cohesive groups in $E$. Since each $1$-cohesive group has
  average satisfaction at least one, its total satisfaction is at
  least $3$. There are $3k$ such groups, so we have that $B$ is at
  least $9k$. Moreover, $B$ is equal to $9k$ exactly if each
  $1$-cohesive group has average satisfaction equal to $1$.  However,
  each committee member is approved by exactly three voters, and each
  of these voters belongs to exactly three $1$-cohesive groups.  Hence
  $B = 9k$ and each $1$-cohesive group has average satisfaction equal
  to $1$.

  Consider some set $S_j = \{u_{j_1}, u_{j_2}, u_{j_3}\}$ such that
  candidate $S_j$ is a member of committee $W$.  Naturally,
  $\{u_{j_1}, u_{j_2}, u_{j_3}\}$ is a $1$-cohesive group, all its
  member approve $S_j$, and, so, its average satisfaction is at least
  $1$.  Indeed, by previous discussion we know that it is exactly $1$.
  Hence, for each voter in $\{u_{j_1}, u_{j_2}, u_{j_3}\}$, candidate
  $S_j$ is the only member of $W$ that he or she approves. If we repeat this
  reasoning for every member of $W$, we find that each of them is
  approved by exactly three voters and no two of them are approved by
  the same voters. This means that $W$ corresponds to an exact cover
  of $U$. The proof is complete.
\end{proof}

\begin{corollary}
  Both \textsc{Constant-PD-Committee} and \textsc{Unit-PD-Committee}
  are $\np$-complete.
\end{corollary}

As all the cohesive groups in the election constructed in the proof of
Theorem~\ref{thm:unit-pd} have cohesiveness level $1$, we have a
stronger result: Given a PD function $f$ such that $f(1) = 1$, it is
$\np$-hard to decide if there is a committee with proportionality
degree $f$.  In particular, we have the next corollary.

\begin{corollary}
  \textsc{Perfect-PD-Committee} is $\np$-hard.
\end{corollary}

We can extend Theorem~\ref{thm:unit-pd} to work for any positive
integer constant $x$ and functions $f$ such that $f(1) = x$.  For
example, for $x = 2$ it suffices to extend the constructed election
with three voters that do not approve anyone and with a single
candidate who is approved by all the other voters.
It would also
be interesting to consider functions $f$ such that $f(1)$ is a
constant between $0$ and $1$, but we leave it for future work.
The above results are nicely aligned with existing polynomial-time
algorithms for computing committees with guarantees on their PD. For
example, there are polynomial-time algorithms for computing EJR
committees, and EJR committees are guaranteed to have PD $f$ such that
$f(\ell) =
\frac{\ell-1}{2}$~\citep{proportional-justified-representation}. As we
see, $f(1) = 0$ (though this could be improved very
slightly\footnote{Since the committee provides EJR, and thus JR, this
  zero could be replaced by $\frac{1}{\nicefrac{n}{k}} = \frac{k}{n}$,
  where $n$ is the number of voters and $k$ is the committee
  size. This follows from the fact that in each $1$-cohesive group of
  size $\frac{n}{k}$ there is at least one candidate who approves at
  least one voter.}).  As we have shown, extending the algorihtm to
find committees with PD functions $f$ such that $f(1) = 1$ (whenever
such committees exist) would not be possible in polynomial time
(assuming $\p \neq \np$).

\section{Computing the PD of a Given Committee}
\label{sec:computing-pd}

Sometimes, instead of computing a committee with a specified PD, we
would like to establish the PD of an already existing one. For
example, this would be the case if we wanted to experimentally compare
how well the committees provided by various voting rules represent the
voters.

One way to proceed would be as follows: For a given election~$E$ and
committee~$W$, consider each cohesiveness level $\ell$ and, using
binary search, find value $f(\ell)$, $0 \leq f(\ell) \leq |W|$, such
that each $\ell$-cohesive group has average satisfaction at least
$f(\ell)$, but for every $\varepsilon > 0$ there exists an
$\ell$-cohesive group with average satisfaction below
$f(\ell)+\varepsilon$ (or there are no $\ell$-cohesive groups in this
election).  Using binary search to compute this value is possible
because in an election with $n$ voters and committee size $k$, there
are at most $O(kn^2)$ different average satisfaction values of
cohesive groups (each cohesive group can have total satisfaction
between $0$ and $nk$, and each cohesive group can have at most $n$
voters).  Running such binary search requires the ability to solve the
following problem.

\begin{definition}
  In the \textsc{PD-Failure} problem we are given an election~$E$, a
  committee~$W$, a cohesiveness level~$\ell$, and a nonnegative
  rational threshold $y \le k$. We ask if $E$ contains an
  $\ell$-cohesive group whose average satisfaction for $W$ is lower
  than $y$.
\end{definition}

As one may expect, this problem is $\np$-complete. Membership in $\np$
follows by nondeterministically guessing an $\ell$-cohesive group and
checking if its average satisfaction is below $y$. To show
$\np$-hardness, we note that setting $y$ to an impossible-to-achieve
value makes the problem equivalent to testing if an $\ell$-cohesive
group exists.

\begin{theorem}\label{thm:pdfailure}
  \textsc{PD-Failure} is $\np$-complete.
\end{theorem}

\begin{proof}
  Membership in $\np$ was already argued in the paragraph above
  the theorem statement.
  To prove $\np$-hardness, we show a reduction from $\cohesivegroup$
  to $\pdf$. Let $(E, k, \ell)$ be an instance of the $\cohesivegroup$
  problem, where $E$ is an election consisting of candidates $C$ and
  voters $V$, $k$ is the size of the final committee, and $\ell$ is
  the cohesiveness level. We ask if there exists an $\ell$-cohesive
  group for election $E$ with committee size~$k$.

  To create a $\pdf$ instance, we use the same election $E$, the same
  committee size $k$, and the same cohesiveness level $\ell$, but we
  add $k$ fresh candidates who are not approved by any voter and
  select them to the committee $W'$.  Further, we set the PD threshold
  $y=k$. From the above observation we conclude that if there exists
  any valid $\ell$-cohesive group, then its average satisfaction with
  $W'$ is $0$, which is strictly less than $y$. Therefore if there
  exists an $\ell$-cohesive group with average satisfaction lower than
  $y$, then this group is also an $\ell$-cohesive group for the
  $\cohesivegroup$ instance. Further, if there are no $\ell$-cohesive
  groups for the $\pdf$ instance, then the $\cohesivegroup$ instance
  doesn't have any $\ell$-cohesive groups either.

  Since the $\cohesivegroup$ problem is $\np$-complete, the $\pdf$
  problem is in $\np$ and we reduced the $\cohesivegroup$ problem to
  the $\pdf$ problem, the $\pdf$ problem is also $\np$-complete.
\end{proof}

\subsection{ILP Formulation}
Fortunately, in practice we may be able to solve instances of our
problem by expressing them as integer linear programs (ILPs) and
solving them using off-the-shelf software.  Specifically, let us
consider an instance of \textsc{PD-Failure} with election~$E = (C,V)$,
committee~$W$, cohesiveness level~$\ell$, and threshold~$y$.  We set
$m = |C|$, $n = |V|$, and $k = |W|$.  For convenience, let $A$ be the
binary matrix of approvals for~$E$, that is, we have $a_{ij} = 1$ if
the $i$-th voter approves the $j$-th candidate, and we have
$a_{ij} = 0$ otherwise.  We note that if there is an $\ell$-cohesive
group $X$ whose satisfaction for $W$ is below $y$, then there is such
a group of size exactly $s = \lceil \ell \cdot \nicefrac{n}{k} \rceil$
(e.g., consider $X$ and remove sufficiently many voters who approve
the most members of~$W$).

To form our ILP instance, we first specify the variables:
\begin{enumerate}
\item For each $i \in [n]$, we have a binary variable $x_i$, with the
  intention that $x_i = 1$ if the $i$-th voter is included in the sought
  cohesive group, and $x_i = 0$ otherwise. 
\item For each $j \in [m]$, we have a binary variable $y_i$, with the
  intention that $y_j = 1$ if all the voters in the group specified by
  variables $x_1, \ldots, x_n$ approve the $j$-th candidate, and
  $y_j = 0$ otherwise.
\end{enumerate}
We refer to the voters (to the candidates) whose $x_i$ ($y_j$) variables
are set to $1$ as \emph{selected}. Next, we specify the
constraints. Foremost, we %
ensure that we select exactly $s$
voters and  at least $\ell$ candidates:
\begin{align*}
  & \textstyle \sum_{i=1}^{n} {x_i} = s, & \text{and}& &\textstyle \sum_{j=1}^{m} {y_j} \ge \ell.
\end{align*}
Then, we %
ensure that each selected voter %
approves all the selected candidates.  For each $j \in [m]$, we form
constraint:
\[
  \textstyle \sum_{i=1}^{n} {a_{ij} \cdot x_i} \ge s \cdot y_j.
\]
If the $j$-th candidate is not selected, then this inequality is
satisfied trivially.  However, if the $j$-th candidate is selected,
then the sum on the left-hand side must be at least $s$, i.e., there
must be at least $s$ selected voters who approve the $j$-th
candidate. Since there are exactly $s$ selected voters, all of them
must approve the $j$-th candidate.

Finally, we %
ensure that the average satisfaction of the
selected voters is below $y$, by adding constraint %
$
  \textstyle \frac{1}{s}\sum_{i=1}^{n} \sum_{j \in W} {a_{ij} \cdot x_i} < y.
$
If there is an assignment %
that satisfies these constraints, then the selected voters form an
$\ell$-cohesive group with average satisfaction below~$y$. Otherwise,
no such group exists.

\subsection{Verification}

For a comparison with previous studies regarding JR and EJR, we also
consider the following verification problem.

\begin{definition}
  In the \textsc{PD-Verification} problem we are given an
  election~$E$, a committee~$W$, a PD function $f$, and we ask if $W$
  has proportionality degree $f$.
\end{definition}

As \textsc{PD-Verification} is very closely related to the complement
of \textsc{PD-Failure}, we find that it is $\conp$-complete (we give the
formal proof in Appendix~\ref{app:pdv}).

\begin{theorem}\label{cor:pdverification}
  \textsc{PD-Verification} is $\conp$-complete.
\end{theorem}

Testing if a committee provides EJR is $\conp$-complete as
well~\citep{justified-representation-in-approval-voting}, so in this
respect PD and EJR are analogous. There is also a polynomial-time
algorithm for testing if a committee provides JR, and in the PD world
this corresponds to a polynomial-time algorithm for checking if a
committee admits a given constant PD function. Such an algorithm was
included as part of the proof of Theorem~\ref{thm:const-pd-np}.

\begin{corollary}
  \textsc{PD-Verification} for a constant PD functions (provided as
  input) is in $\p$.
\end{corollary}

\section{Dealing With Computational Hardness}

In this section we consider circumventing the computational hardness
of our problems by studying their parameterized complexity and by
considering structured elections.

\subsection{Fixed-Parameter Tractability}
\label{sec:fpt}

Our two main problems, \textsc{PD-Committee} and \textsc{PD-Failure},
are fixed-parameter tractable with respect to the number of candidates
and the number of voters.

For \textsc{PD-Failure} and the parameterization by the number of
candidates, we proceed similarly as in the proof of
Theorem~\ref{thm:const-pd-np}. Namely, for each set $R$ of candidates
we consider the set $V(R)$ of all the voters that approve members of
$R$, one-by-one remove from this set the voters with the highest
satisfation, and watch if at any point we obtain an $\ell$-cohesive
group with average satisfaction below the required value. Using a
similar approach, and trying every possible committee, we also obtain
an algorithm for \textsc{PD-Committee}.

For the parameterization by the number of voters, we solve our
problems by forming ILP instances and solving them using the classic
algorithm of \citet{len:j:integer-fixed}. This is possible because
with $n$ voters there are at most $2^n$ cohesive groups and each
candidate has one of $2^n$ types (where the type of a candidate is the
set of voters that approve him; candidates with the same type are
interchangeable).

\begin{theorem}\label{thm:fpt-thm}
  There are FPT algorithms for \textsc{PD-Committee} and
  \textsc{PD-Failure} both for the parameterization by the number of
  candidates and for the parameterization by the number of voters.
\end{theorem}

\begin{proof}
  Let us first consider the parameterization by the number of
  candidates and the \textsc{PD-Failure} problem. Our input consists
  of an election $E = (C,V)$, committee $W$ of size $k$, cohesiveness
  level~$\ell$, and rational threshold $y$. Let $m$ be the number of
  candidates and let $n$ be the number of voters. For each subset of
  $\ell$ candidates, we find a group of $\ell \cdot \frac{n}{k}$
  voters who are least satisfied with~$W$. If the lowest satisfaction
  among such groups is below $y$ then we accept and otherwise we
  reject. The correctness and fixed-parameter tractability follow
  immediately.

  For parameterization by the number of candidates and the
  \textsc{PD-Committee} problem, it suffices to try all committees and
  for each of them (and each cohesiveness level) use the algorithm for
  \textsc{PD-Failure} to check if it indeed achieves required PD.

  Next let us move on to the parameterization by the number of voters
  and the \textsc{PD-Failure} problem. We use the same notation as in
  the argument above for parameterization by the number of
  candidates. It suffices to consider every subset of voters, check if
  it is an $\ell$-cohesive group, and verify if its average
  satisfaction is below $y$.

  For the case of \textsc{PD-Committee} and parameterization by the
  number of voters, we employ integer linear programming. Let
  $E = (C,V)$ be the input election with $m$ candidates and $n$
  voters. We seek a committee of size $k$, with PD $f$.  There are
  $2^n$ subsets of the voters, and we order them in some way, so for
  each $i \in [2^n]$ we can speak of the $i$-th subset. For each such
  subset, we say that a candidate has type $i$ if she is approved
  exactly by the voters from the $i$-th subset (and only by them). For
  each $i \in [2^n]$ we let $c_i$ be the number of type-$i$ candidates
  and we form a variable~$x_i$, with the intended meaning that $x_i$
  is the number of type-$i$ candidates in the committee. We introduce
  the following constraints:
  \begin{enumerate}
  \item For each $i$, we require that $x_i \leq c_i$, i.e., that we do not
    select more type-$i$ candidates than available.
  \item We require that $\sum_{i \in [2^n]} x_i = k$, i.e., we ensure
    that we select a committee of size exactly $k$.
  \item For each $\ell \in [k]$ and each $\ell$-cohesive group $S$ of
    voters (due to our parameterization, we can enumerate them all),
    we form the following constraint:
    \[
      \sum_{v \in S} \sum_{i \in [2^n]} x_i \cdot [\text{$v$ approves type-$i$ candidates}] \geq |S| \cdot f(\ell),
    \]
    where we use the Iverson bracket notation (i.e., for a true/false
    statement $F$, by $[F]$ we mean~$1$ is $F$ is true and we mean $0$
    otherwise). This constraint ensures that each cohesive group has
    required level of average satisfaction.
  \end{enumerate}
  We solve this ILP instance using the classic algorithm of
  \citet{len:j:integer-fixed}. Since the number of variables is $2^n$
  and $n$ is the parameter, doing so is possible in FPT time. This
  completes the proof.
\end{proof}

Testing if a committee provides EJR is also fixed-parameter tractable
for the parameterizations considered in Theorem~\ref{thm:fpt-thm}. So, from
this point of view, dealing with PD is not harder than dealing with
EJR.

Finally, the problem of counting cohesive groups (and, thus, also the
problem of deciding if groups with particular cohesiveness level
exist) also is fixed-parameter tractable for our parameters. For
parameterization by the number of candidates, we can use the
inclusion-exclusion principle, and for the parameterization by the
number of voters we can explicitly look at each subset of voters.

\begin{theorem}
  There are FPT algorithms for \textsc{(\#)Cohesive-Group}, for the
  parameterizations by the number of candidates and by the number of
  voters.
\end{theorem}

\begin{proof}
  Let us first consider the parameterization by the number of voters
  and the \textsc{\#Cohesive-Group} problem. Our input consists of an
  election $E = (C,V)$, committee $W$ of size $k$, and cohesiveness
  level $\ell$. Let $m$ be the number of candidates and let $n$ be the
  number of voters. Initially, we have a counter set to zero. For each
  subset of at least $\ell \cdot \frac{n}{k}$ voters, we compute the
  set of candidates that are approved by all these voters. If this set
  has size at least $\ell$ then we increase the counter and otherwise
  we do not. At the end, the counter contains the desired answer. For
  the \textsc{Cohesive-Group} problem it is enough to check if the
  counter is above $0$.
  
  For the parameterization by the number of candidates and the
  \textsc{Cohesive-Group} problem, for each subset of $\ell$
  candidates we calculate the number of voters that approve all these
  candidates and accept if it is at least $\ell \cdot \frac{n}{k}$, we
  reject if we do not accept for any subset. For the
  \textsc{\#Cohesive-Group} problem, we can use the
  inclusion-exclusion principle.
  
\end{proof}

\subsection{Structured Preferences}
\label{sec:domains}

Next we consider two domains of structured preferences, introduced by
\citet{structure-in-dichotomous-preferences}. Such domains are
interesting because, on the one hand, they capture some realistic
scenarios, and, on the other hand, by assuming them it is often
possible to provide polynomial time algorithms for problems that in
general are intractable.

\begin{definition}[\textbf{\citet{structure-in-dichotomous-preferences}}]
  An election $E = (C,V)$ has candidate interval (CI) preferences
  (voter interval preferences, VI) if it is possible to order the
  candidates (the voters) so that for each voter $v$ (for each
  candidate $c$) the set $A(v)$ (the set $A(c)$) is an interval w.r.t. this order.
\end{definition}
For an example of CI preferences, consider a political election
where candidates are ordered according to the left-to-right spectrum
of opinions and the voters approve ranges of candidates whose opinions
are close enough to their own.
\citet{structure-in-dichotomous-preferences} gave algorithms for
deciding if a given election has CI or VI preferences, and for
computing appropriate orders of candidates or voters. Thus, for
simplicity, we assume that these orders are provided together with our
input elections.
We mention that a number of other preference domains are considered in
the literature---see, e.g., the works of \citet{yang2019tree} and
\citet{GBSF21}---but the CI and VI ones are by far the most popular.
For a very detailed discussion of structured domains, albeit in the
world of ordinal preferences, we point to the survey of
\citet{elk-lac-pet:t:ordinal-restricted-domains}.

Unfortunately, even for CI and VI elections we do not know how to
solve the \textsc{PD-Committee} problem in
polynomial-time.\footnote{In particular, the approach of
  \citet{pet-lac:j:spoc} based on solving totally unimodular ILP
  instances does not seem to work here.}
Nonetheless, we do have polynomial-time
algorithms for the \textsc{PD-Failure} problem.

\begin{theorem}
   \textsc{PD-Failure} restricted to either CI or VI elections is in $\p$.
\end{theorem}

\begin{proof}
  First, we give an algorithm for the CI case. Our input consists of
  an election $E = (C,V)$, where $C = \{c_1, \ldots, c_m\}$ and
  $V = (v_1, \ldots, v_n)$, a size-$k$ committee $W$, cohesiveness
  level $\ell$, and threshold value $y$. Without loss of generality,
  we assume that $E$ is CI with respect to the order
  $c_1 \lhd c_2 \lhd \cdots \lhd c_m$.

  Since $E$ is a CI election, we observe that if $X$ is some cohesive
  group whose all members approve some two candidates $c_i$ and $c_j$,
  $i \leq j$, then all members of $X$ also approve candidates
  $c_{i+1}, \ldots, c_{j-1}$. For each $i \leq m-\ell+1$, let $X(i)$
  be the set of all voters who approve each of the candidates
  $c_{i}, \ldots, c_{i+\ell-1}$. By the preceding argument, we see
  that every $\ell$-cohesive group can be obtained by taking some set
  $X(i)$ and (possibly) removing some of its members.

  Our algorithm proceeds as follows. For each set $X(i)$, we form a
  set $Y(i)$ by taking $X(i)$ and removing all but
  $\lceil \ell \cdot \nicefrac{n}{k} \rceil$ voters that are least
  satisfied with $W$. (If a given $X(i)$ contains fewer than
  $\lceil \ell \cdot \nicefrac{n}{k} \rceil$ voters then we set
  $Y(i) = \emptyset$ and we assume that the average satisfaction of
  its voters is $+\infty$.)
  If there is some $i$ such that the average satisfaction of the
  voters in $Y(i)$ is below $y$, then we accept (indeed, we have just
  found an $\ell$-cohesive group with average satisfaction below
  $y$). If there is no such $Y(i)$, then we reject (we do so because
  each nonempty $Y(i)$ has the lowest average satisfaction among all
  the $\ell$-cohesive groups that can be obtained by removing voters
  from $X(i)$).  Correctness and polynomial running time follow
  immediately.

  Now let us consider the VI case. We use the same notation as before,
  except that we assume that $E$ is VI with respect to the voter order
  $v_1 \lhd v_2 \lhd \cdots \lhd v_n$. We use the same algorithm as in
  the CI case, but for the $X(i)$ sets defined as follows (let
  $s = \lceil \ell \cdot \nicefrac{n}{k}\rceil$):
  For each $i \in [n-s+1]$, we let
  $X(i) = \{v_{i}, v_{i+1}, \ldots, v_{j}\}$, where $j$ is the largest
  value such that $|A(v_{i}) \cap A(v_{j})| \geq \ell$ (if $v_i$
  approves fewer than $\ell$ candidates then $X(i)$ is empty).
  The algorithm remains correct because, as in the CI case, every
  $\ell$-cohesive group is a subset of some $X(i)$.
\end{proof}

Similar reasoning and observations as in the above proof also give the
algorithms for counting cohesive groups (and, thus, for deciding their
existence).
\begin{theorem}\label{thm:civi-counting}
  \textsc{(\#)Cohesive-Group} restricted to either CI or VI elections
  is in $\p$.
\end{theorem}

\noindent We prove Theorem~\ref{thm:civi-counting} via the following
two theorems (they suffice due to Proposition~\ref{red:12}).

\begin{theorem}
  There is a polynomial-time algorithm for the $\sharpfscg$ problem
  under the VI restriction
\end{theorem}

\newcommand{\smallestCG}{\mathit{smallestCG}}

\begin{proof}
  Let $(E, k, \ell, x)$ be a $\sharpfscg$ VI instance, where $E$ is an
  election with candidates $C$ and voters $V$, $k$ is the committee
  size, $\ell$ is the cohesiveness level, and $x$ is the size of
  cohesive groups. We ask how many $\ell$-cohesive groups of size $x$
  are there in election $E$.  We assume that $V = (v_1, \ldots, v_n)$
  and the election is VI for this order of the candidates.

  We observe that if voters $v_i$ and $v_j$ approve candidate $c_p$,
  then each voter $v_k$ between $v_i$ and $v_j$ also approves $c_p$,
  because under VI each candidate is approved by a consecutive segment
  of voters. As a result, if $v_i$ and $v_j$ have at least $\ell$
  common candidates, then each voter $v_k$ between $v_i$ and $v_j$
  also approves these candidates.

  By $\smallestCG(v_i, \ell, x)$ we mean the number of $\ell$-cohesive
  groups of size $x$ in which voter~$v_i$ has the smallest
  index. Then, the sum of the $\smallestCG$ values over all the voters
  is the final answer. Now let us show how to calculate
  $\smallestCG(v_i, \ell, x)$.

  Given a voter $v_i$, a group cohesiveness level $\ell$, and an
  integer $x$, we find the greatest index $j$ such that voter $v_j$
  still has at least $\ell$ common candidates with $v_i$. If $v_j$
  does not exist or the number of voters in range $[v_i,v_j]$ is lower
  than $x$, then return $0$. Otherwise, we select the voter $v_i$ and
  $x-1$ other voters from $[v_{i+1},v_j]$; we can do it in
  ${j-i} \choose {x-1}$ ways and this is the value we output.  This
  completes the proof.
\end{proof}

\begin{theorem}
  There is a polynomial-time algorithm for the $\sharpfscg$ problem
  under the CI restriction
\end{theorem}

\begin{proof}
  Let $(E, k, \ell, x)$ be a $\sharpfscg$ CI instance, where $E$ is an
  election with candidates $C$ and voters $V$, $k$ is the committee
  size, $\ell$ is the cohesiveness level, and $x$ is the size of
  cohesive groups. We ask how many $\ell$-cohesive groups of size $x$
  are there in election $E$. We assume that $C = \{c_1, \ldots, c_m\}$
  and the election is CI for candidate order $c_1,c_2, \ldots, c_m$.

  We observe that if candidates $c_i$ and $c_j$ are approved by voter
  $v_p$, then each candidate $c_k$ between $c_i$ and $c_j$ is also
  approved by $v_p$, because under CI each voter approves a
  consecutive segment of candidates.  As a result, for each
  $\ell$-cohesive group the candidates approved by all its members
  form a consecutive segment.

  By $\smallestCG(c_j, \ell, x)$ we mean the number of $\ell$-cohesive
  groups of size $x$ in which candidate $c_j$ is the commonly approved
  candidate with the smallest index. Then, the sum of the
  $\smallestCG$ values through all the candidates is the final
  answer. Now let us show how to calculate the function
  $\smallestCG(c_j, \ell, x)$.

  Assume we are given candidate $c_j$, group cohesiveness level
  $\ell$, and an integer $x$. Let $L_1$ be the set of voters that
  approve all the candidates from
  $\{c_j, c_{j+1}, ... , c_{j+\ell-1}\}$ and at least one candidate
  which has index strictly smaller than $j$. Similarly, let $L_2$ be
  the set of voters that approve all the candidates from
  $\{c_j, c_{j+1}, ... , c_{j+\ell-1}\}$ and do not approve any
  candidate whose index is strictly smaller than $j$. Both values can
  be calculated in polynomial-time by a single iteration through
  election $E$. Now let us point out that each $\ell$-cohesive group
  accounted for in $\smallestCG(c_j, \ell, x)$ must consist of at
  least one voter included in $L_2$ and some voters included in
  $L_1$. As we do not know how many voters we should take from the
  first part, we iterate through all possible partition sizes. Thus,
  the number of $\ell$-cohesive groups of size $x$ whose members'
  smallest index of a commonly approved candidate is $j$, is equal to:
  \[
    \textstyle \sum_{k=1}^{\min(|L_2|,x)} { { |L_2| \choose k } \cdot { |L_1| \choose
        {x-k} } }
  \]
  This completes the proof.
\end{proof}

Similar approach shows that testing if a committee provides EJR can be
done in polynomial time for CI or VI elections (to the best of our knowledge, this is a folk
result).

\paragraph{Perfect PD in CI/VI Elections?}

\citet{complexity-of-ejr-and-pjr} have shown that for each election
and each committee size there is a committee with a nearly perfect PD,
but there are scenarios where committees with perfect PDs do not
exist. Unfortunately, this remains true even if the elections are CI
and VI at the same time.

\begin{example}
  Consider an election $E = (C,V)$, where $C =\{c_1, \ldots, c_7\}$,
  and $V = (v_1, \ldots, v_{15})$. The committee size is $k = 5$ and
  the approval sets are as follows:
  \begin{center}
    \scalebox{0.9}{
    \begin{tabular}{c|ccccccc}
      \toprule
              & $c_1$ & $c_2$& $c_3$& $c_4$& $c_5$& $c_6$& $c_7$\\
      \midrule
      $v_1$   &   1   &   -  &   -  &  -  &   -  &   -  &  -   \\
      $v_2$   &   1   &   -  &   -  &  -  &   -  &   -  &  -   \\
      $v_3$   &   1   &   1  &   -  &  -  &   -  &   -  &  -   \\
      $v_4$   &   -   &   1  &   -  &  -  &   -  &   -  &  -   \\
      $v_5$   &   -   &   1  &   1  &  -  &   -  &   -  &  -   \\ 
      $v_6$   &   -   &   -  &   1  &  -  &   -  &   -  &  -   \\
      $v_7$   &   -   &   -  &   1  &  1  &   -  &   -  &  -   \\
      $v_8$   &   -   &   -  &   -  &  1  &   -  &   -  &  -   \\
      $v_9$   &   -   &   -  &   -  &  1  &   1  &   -  &  -   \\
      $v_{10}$ &   -   &   -  &   -  &  -  &   1  &   -  &  -   \\
      $v_{11}$ &   -   &   -  &   -  &  -  &   1  &   1  &  -   \\
      $v_{12}$ &   -   &   -  &   -  &  -  &   -  &   1  &  -   \\
      $v_{13}$ &   -   &   -  &   -  &  -  &   -  &   1  &  1   \\
      $v_{14}$ &   -   &   -  &   -  &  -  &   -  &   -  &  1   \\
      $v_{15}$ &   -   &   -  &   -  &  -  &   -  &   -  &  1   \\
      \bottomrule                                                     
    \end{tabular}}
  \end{center}
  Clearly, the election is both CI and VI. We see that
  $\nicefrac{n}{k} = 3$ and, thus, for each $i \in [7]$, voters
  $v_{2i-1}, v_{2i}, v_{2i+1}$ form a cohesive group (for
  candidate~$c_i$).

  Now consider some size-$k$ committee. If it does not contain some
  candidate $c_i$, then the $1$-cohesive group
  $\{v_{2i-1}, v_{2i}, v_{2i+1}\}$ must have average satisfaction
  below $1$. Indeed, altogether members of this group give at most
  five approvals, of which three go to $c_i$. Thus, without $c_i$, the
  average satisfaction is at most $\frac{2}{3} < 1$.  However, since
  the committee size is five and there are seven candidates, for each
  committee there is some $1$-cohesive group with satisfaction below
  $1$.  Thus there is no committee with a perfect PD for this election
  and committee size five.
\end{example}

\section{Conclusions and Future Work}
We have shown that computing committees with a given proportionality
degree is, apparently, more difficult than computing EJR committees,
but verification problems for these two notions have the same
complexity. Two most natural directions of future work would be to
establish the exact complexity of the \textsc{PD-Committee} problem
and experimentally analyze PDs of committees provided by various
voting rules.

\section*{Acknowledgments}
This project has received funding from the European 
    Research Council (ERC) under the European Union’s Horizon 2020 
    research and innovation programme (grant agreement No 101002854).
    
\noindent \includegraphics[width=3cm]{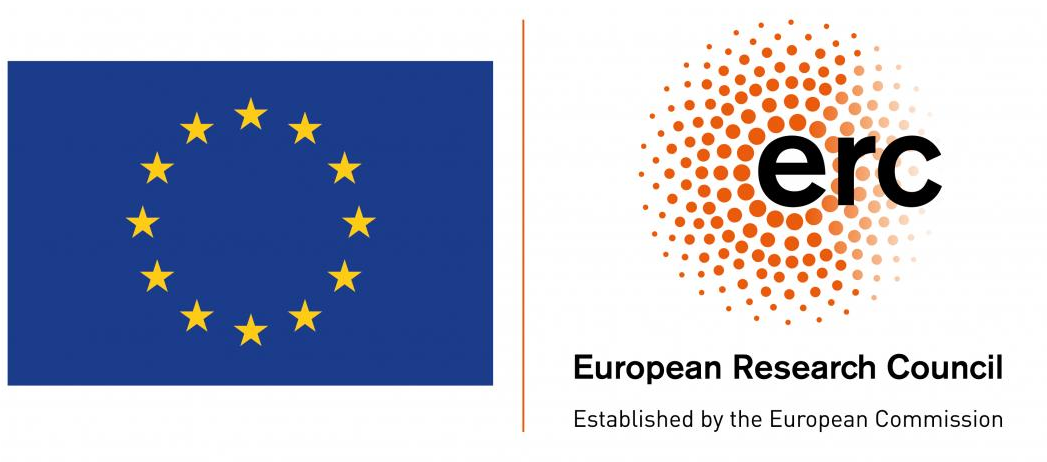}

\bibliography{bibliography}

\appendix

\section{Proof of Proposition~\ref{red:12}}
\label{app:counting}

We prove Proposition~\ref{red:12} via Lemmas~\ref{red:1}
and~\ref{red:2} below.

\begin{lemma}\label{red:1}
    \textsc{\#Cohesive-Group} $\le_{\mathrm{T}}^{\mathrm{fp}}$ \textsc{\#FSCG}
\end{lemma}

\begin{proof}
  Let $(E, k, \ell)$ be an instance of \textsc{\#Cohesive-Group},
  where $E$ is an election instance with $m$ candidates and $n$
  voters, $k$ is the committee size, and $\ell$ is the cohesiveness
  level.

  To count the number of $\ell$-cohesive groups, it suffices to sum up
  the number of $\ell$-cohesive groups of each possible size $x$. From
  the definition of an $\ell$-cohesive group, we know that its size
  must be at least $\ell \cdot \frac{n}{k}$. It is also clear that its size
  cannot exceed the number of voters. Thus,
  $\sharplcg(E, \ell, k) = \sum_{x = \lceil \ell \cdot \frac{n}{k} \rceil}^{n}
  {\sharplcg(E, \ell, k, x)}$. This concludes the argument.
\end{proof}

We also have a reduction in the reverse direction, thus obtaining
computational equivalence between \textsc{\#Cohesive-Group} and
\textsc{\#FSCG}.

\begin{lemma}\label{red:2}
  \textsc{\#FSCG} $\le_{\mathrm{T}}^{\mathrm{fp}}$ \textsc{\#Cohesive-Group}
\end{lemma}
\begin{proof}
  Let $(E, k, \ell, x)$ be an instance of \textsc{\#FSCG}, where $E$
  is an election with $m$ candidates and $n$ voters, $k$ is the size
  of the final committee, $\ell$ is the cohesiveness level, and $x$ is
  the size of the considered groups. We assume that
  $\ell \cdot \frac{n}{k} \le x \le n$ and $1 \le \ell \le m$ as
  otherwise we would immediately output zero as the answer.

  We create an election $E'$ which, initially, is a copy of $E$.
  Next, we extemed $E'$ by adding $n \cdot x \cdot (x+1) - n$ new
  voters that do not approve any candidates and
  $n \cdot x \cdot (x+1) \cdot m - m$ new candidates that are not
  approved by any voter. Thus, in $E'$ we have
  $n' = n \cdot x \cdot (x+1)$ voters and
  $m' = n \cdot x \cdot (x+1) \cdot m$ candidates. The aim of adding
  the new voters and candidates is to establish the lower bound on the
  size of the cohesive groups.

  If we were able to count the number of $\ell$-cohesive groups with
  the size greater or equal to $x$ and those with the size strictly
  greater than $x$, then the difference between these two values would
  be the number of $\ell$-cohesive groups with size $x$.

  To use the idea from the preceding paragraph, we define
  $k_1^{'} = \ell \cdot n \cdot x$ and
  $k_2^{'} = \ell \cdot n \cdot (x+1)$. It should be clear that
  $1 \le k_1^{'} < k_2^{'} \le m'$. Further, we note that
  $\ell \cdot \frac{n'}{k_1^{'}} = \ell \cdot \frac{n \cdot x \cdot
    (x+1)}{\ell \cdot n \cdot x} = x+1$ and, so,
  $\sharplcg(E', k_1^{'})$ is equal to the number of $\ell$-cohesive
  groups in $E'$ with size at least $x+1$. Similarly, as
  $\ell \cdot \frac{n'}{k_2^{'}} = \ell \cdot \frac{n \cdot x \cdot
    (x+1)}{\ell \cdot n \cdot (x+1)} = x$, we have that
  $\sharplcg(E', k_2^{'})$ is equal to the number of $\ell$-cohesive
  groups in $E'$ with size at least $x$. Furthermore, as newly created
  voters in $E'$ do not approve any candidates and newly created
  candidates in $E'$ are not approved by any voter, each
  $\ell$-cohesive group in $E'$ consists of the voters from $E$
  approving only the candidates from $E$, so it is also a valid
  $\ell$-cohesive group in~$E$. Thus,
  $\sharplcg(E, k, \ell) = \sharplcg(E', k_1^{'},\ell) - \sharplcg(E, k_2^{'},\ell)$.

  This completes the proof as we have shown a polynomial-time
  algorithm that solves \textsc{\#FSCG} using oracle access to
  \textsc{\#Cohesive-Group}.
\end{proof}

\section{Proof of Theorem~\ref{cor:pdverification}}
\label{app:pdv}

\newtheorem*{verification}{Theorem~\ref{cor:pdverification}}

\begin{verification}
  \textsc{PD-Verification} is $\conp$-complete.
\end{verification}

\begin{proof}
  Let us show that $\overline{\pdv}$ is in $\np$.  Given a
  $\overline{\pdv}$ instance, we guess a group of voters and a
  cohesiveness level $\ell$. We can verify in polynomial time whether
  these voters form an $\ell$-cohesive group. Then we calculate the
  average satisfaction of the group and compare it with the given
  threshold. If the voters form an $\ell$-cohesive group and their
  average satisfaction is lower than the threshold, then the selected
  voters witness $\overline{\pdv}$. Therefore $\overline{\pdv}$ is in
  $\np$ and $\pdv$ is in $\conp$.

  We show a reduction from the $\pdf$ problem to the complement of the
  $\pdv$ problem.  Let $(E, W, k, \ell, y)$ be a $\pdf$ instance,
  where $E$ is an election, $W$ is a final committee of size $k$,
  $\ell$ is the cohesiveness level, and $y$ is a nonnegative real
  threshold $y \le k$. We ask if there exists an $\ell$-cohesive group
  whose average satisfaction is lower than $y$.

  We create a $\pdv$ instance as follows. We have the same election
  $E$ and the same committee $W$. We set the PD function to be
  $f(\ell) = y$, and $0$ otherwise.

  Suppose that the answer for the $\pdv$ instance is ``no''. Then, for
  some $\ell'$ there exists an $\ell'$-cohesive group $S$ whose
  average satisfaction is lower than $f(\ell')$. It is quite clear
  that each $\ell'$-cohesive group has average satisfaction at least
  $0$, regardless of a selected committee. It means that $\ell'$ must
  be equal to $\ell$. Therefore $S$ is an $\ell$-cohesive group and
  has average satisfaction lower than $f(\ell') = f(\ell) = y$. Thus
  the answer for the $\pdf$ instance is ``yes''.

  Suppose that the answer for the $\pdv$ instance is ``yes''. Then,
  for each $\ell'$ and each $\ell'$-cohesive group its average
  satisfaction is at least $f(\ell')$. In particular, it means that
  each $\ell$-cohesive group has average satisfaction at least
  $f(\ell) = y$. Therefore there does not exist an $\ell$-cohesive
  group that has average satisfaction lower than $y$. Thus the answer
  for the $\pdf$ instance is ``no''.

  Since the $\pdv$ problem is in $\conp$, the $\pdf$ problem is
  $\np$-complete, and we reduced the $\pdf$ problem to the complement
  of the $\pdv$ problem, the $\pdv$ problem is $\conp$-complete.
\end{proof}

\end{document}